\documentclass[11pt]{article}
\usepackage[utf8]{inputenc}
\usepackage{amsmath,amsfonts,amsthm,amssymb}
\usepackage{array}
\usepackage{url}
\usepackage{graphicx}
\usepackage{algpseudocode}
\usepackage{bbm}
\usepackage{float}
\usepackage{framed}
\usepackage{enumerate}
\usepackage{enumitem}
\usepackage{tabularx,ragged2e,booktabs,caption}
\usepackage{color}
\usepackage[colorlinks=true, linkcolor=red, urlcolor=blue, citecolor=blue]{hyperref}
\usepackage{fancyvrb}
\usepackage{ bbold }
\usepackage{commath}

\newtheorem{theo}{Theorem}[section]
\newtheorem{definition}{Definition}[section]

\newtheorem{lemma}[theo]{Lemma}

\newcommand{\II}{{\cal I}}

\newcommand{\E}[1]{{\mathbb E}\left[ #1 \right]}
\DeclareMathOperator*{\argmin}{arg\,min}

\title{Efficient Splitting of Measures and Necklaces}
\author{
Noga Alon
\thanks
{Department of Mathematics, Princeton University,
Princeton, NJ 08544, USA.
Email: {\tt nalon@math.princeton.edu}.
Research supported in part by
NSF grant DMS-1855464
and the Simons Foundation.
}
\and
Andrei Graur
\thanks
{Department of Mathematics, Princeton University,
Princeton, NJ 08544, USA.
Email: {\tt agraur@princeton.edu}. 
}}
\date{}

\begin{document}

\maketitle

\begin{abstract}

   We provide approximation algorithms for two problems, known as
   NECKLACE SPLITTING and $\epsilon$-CONSENSUS SPLITTING. In the problem
   $\epsilon$-CONSENSUS SPLITTING, there are $n$ 
non-atomic probability measures on
   the interval $[0, 1]$ and $k$ agents. The
   goal is to divide the interval, via at most $n (k-1)$ cuts, into pieces
   and distribute them to the $k$ agents in an approximately equitable
   way, so that the discrepancy between the shares of any two agents,
   according to each measure,
   is at most $2 \epsilon / k$. It is known that this is possible
   even for $\epsilon = 0$. NECKLACE SPLITTING is a discrete version of
   $\epsilon$-CONSENSUS SPLITTING. For $k = 2$ and some absolute positive
   constant $\epsilon$, both of these problems are PPAD-hard.

   We consider two types of approximation. The first  
   provides every agent a positive amount of measure of each type
   under the constraint of making at most $n (k - 1)$ cuts. The second
   obtains an approximately equitable split with as
   few cuts as possible. Apart from the offline model, we consider the
   online model as well, where the interval (or necklace) is presented
   as a stream, and decisions about cutting and distributing must be
   made on the spot.

   For the first type of approximation, we describe an efficient algorithm
   that gives every agent at least $\frac{1}{nk}$ of each 
   measure and works even online. For the second type
   of approximation, we provide an efficient online algorithm that makes
   $\text{poly}(n, k, \epsilon)$ cuts and an offline algorithm making
   $O(nk \log \frac{k}{\epsilon})$ cuts. We also establish lower
bounds for the number of cuts required
   in the online model for both problems even 
   for $k=2$ agents, showing that the number of cuts
   in our online algorithm is optimal up to a logarithmic factor. 
\end{abstract}

\section{Introduction}

\subsection{The problems}
The $\epsilon$-Consensus Splitting problem deals with a
fair partition of an interval among $k$ agents, according to $n$
measures. Necklace Splitting is a discrete version of the problem
where the objective is to cut a necklace
with beads of $n$ colors into intervals and distribute them 
to $k$ agents in an equitable way. Both 
problems can be solved using at most $n (k-1)$ cuts, as shown in
\cite{Al}. The proofs apply topological
arguments and are non-constructive. See also \cite{LZ} and \cite{SS}
for two and three-dimensional versions of the results. 
Known hardness results discussed in subsection 1.2, have
been proved for the original versions of these two
problems. These suggest pursuing the challenge of finding efficient
approximation
algorithms, as well as that of 
proving non-conditional hardness
in restricted models. Before adding more on the background,
we give the formal definitions of the two problems.

\begin{definition}
\label{d11}
\textbf{($\epsilon$-Consensus Splitting)} An instance $I_{n, k}$ of
$\epsilon$-Consensus Splitting with $n$ measures and $k$ agents consists
of $n$ non-atomic probability measures on the interval $[0, 1]$, which we
denote by $\mu_i$, for $i \in [n]=\{1,2,, \ldots ,n\}$. 
The goal is to split the interval, via
at most $n (k - 1)$ cuts, into subintervals and distribute them
to the $k$ agents so that for every two agents $a,b \in [k]$ 
and every measure
$i \in [n]$, we have $|\mu_i(U_a) - \mu_i(U_b)| \le \frac{2\epsilon}{k}$,
where $U_a, U_b$ are the unions of all intervals $a, b$ receive, respectively.
\end{definition}

For any allocation of the interval $[0, 1]$ to the
$k$ agents,  define the \textit{absolute discrepancy} as $\max_{a,
b \in [k], i \in [n]} |\mu_i(U_a) - \mu_i(U_b)|$. This is
the maximum difference, over all measures,
between the shares of two distinct agents.
A valid solution for the
$\epsilon$-Consensus Splitting problem is thus a set of at most
$n (k - 1)$ cuts on $[0, 1]$ and a partition 
of the resulting intervals among the $k$
agents so that the absolute discrepancy is at most $2 \epsilon / k$. 
The notion of absolute discrepancy for the necklace
problem is defined analogously. \\

The $\epsilon$-Consensus Splitting problem has a solution
for every instance, even if $\epsilon = 0$, as proved in
\cite{Al}. The proof is
non-constructive, that is, it does not yield 
an efficient algorithm for producing the cuts and the partition
for a given input. The 
$\epsilon$-Consensus Splitting
problem was first mentioned more than 70 years ago 
in \cite{Ne}. In \cite{SS} it is called
the Consensus-$1/k$-division problem. \\

For $k = 2$, the problem is closely related to the
Hobby-Rice Theorem \cite{HR}. In \cite{FG}, as well as in \cite{FG1}
and \cite{FG2}, Filos-Ratsikas, Goldberg and their collaborators
consider the
case $k = 2$. They call this version of the problem 
the \textit{$\epsilon$-Consensus Halving} problem, a terminology that
we adapt here. Some of the results we present are
proved only for $\epsilon$-Consensus Halving, but can be
generalized for
$\epsilon$-Consensus Splitting, as discussed in Section 6.

\begin{definition}
\label{d12}
\textbf{(Necklace Splitting)} An instance of Necklace Splitting for $n$
colors and $k$ agents consists of a set of beads ordered along
a line, where each bead is colored by exactly one color $i
\in [n]=\{1,2, \ldots ,n\}$. 
The goal is to split the necklace, via at most $n (k - 1)$
cuts made between consecutive beads into intervals and distribute
them
to the $k$ agents so that for each color $i$, every agent gets
either $\lceil \frac{m_i}{k} \rceil$ or $\lfloor \frac{m_i}{k} \rfloor$
beads of color $i$, where $m_i$ is the number of beads of color $i$.
\end{definition}

Note that this definition is slightly 
broader than the one given in
\cite{Al}, where it is assumed that $m_i$ is divisible
by $k$ for all $i \in [n]$.
However, as shown in \cite{AM},
these two forms of the Necklace Splitting problem are equivalent. 
As in the case
with $\epsilon$-Consensus Splitting, we call the special case
$k=2$ of two agents the
Necklace Halving problem. \\

The existence of a solution for the Necklace Splitting problem was proved,
using topological arguments, first for $k = 2$ agents in \cite{GW} (see
also \cite{AW} for a short proof), and 
then for the general case of $k$ agents
in \cite{Al}. A more recent combinatorial proof of this existence
result appears in \cite{Me}.
As in the case with $\epsilon$-Consensus
Splitting, these proofs are non-constructive. 
The Necklace Halving problem is first discussed
in \cite{BL}. The problem of finding an efficient 
algorithmic proof of
Necklace Splitting is mentioned in \cite{Al2}.  \\

Recently, there have been several results regarding the hardness of the
$\epsilon$-Consensus Halving and the 
Necklace Halving problems. 
These are discussed in the next subsection.

\subsection{Hardness and Approximation}
PPA and PPAD are two complexity classes introduced in the seminal paper
of Papadimitriou, \cite{Pa}. Both of these are contained in the class
TFNP, which is the complexity class of \textit{total search} problems,
consisting of all problems in NP where a solution exists for every
instance. A problem is PPA-complete if and only if it is
polynomially equivalent to the canonical problem LEAF, described in
\cite{Pa}. Similarly, a problem is PPAD-complete if and only
if it is polynomially equivalent to the problem END-OF-THE-LINE. 
A problem is PPA-hard or PPAD-hard if the respective canonical problem
is polynomially reducible to it. A number of important problems, such as
several versions of Nash Equilibrium \cite{DW} and Market Equilibrium
\cite{CS}, have been proved to be PPAD-complete. It is known
that PPAD $\subseteq$ PPA. Hence, PPA-hardness implies PPAD-hardness. \\

Filos-Ratsikas and Goldberg showed that the
$\epsilon$-Consensus Halving problem, first for $\epsilon$ inversely
exponential \cite{FG1} then for $\epsilon$ inversely polynomial 
in the number of measures \cite{FG},
as well as Necklace Halving, are PPA-hard problems. Furthermore, in
a subsequent paper with Frederiksen and 
Zhang \cite{FG2}, they showed that there exists a constant $\epsilon >
0$ for which $\epsilon$-Consensus Halving is PPAD-hard. 
Our
main objective here is to find efficient approximation algorithms for
these problems.

\subsection{Our contribution}
We consider approximation algorithms for two versions of the
problem, which we call
\textit{type 1 approximation} and
\textit{type 2 approximation}, respectively. The first one is in the
context of the $\epsilon$-Consensus Splitting problem with the aim of
providing strictly positive measure of each type to
each of the $k$ agents using
at most $n (k - 1)$ cuts. In the second one,
for both $\epsilon$-Consensus Halving and Necklace Halving, we allow
the algorithm to make more than $n$ cuts, and expect a \textit{proper}
solution. A \textit{proper} solution is a finite set of cuts
and a distribution of the resulting intervals to the $k$ agents
so that the absolute discrepancy is at most $2 \epsilon/k$ 
(or at most $ 1$ in the
case of Necklace Splitting). The objective is to minimize the number
of cuts the algorithm makes. Type 2 approximation
has been considered earlier in \cite{BL2} and \cite{BT}.  \\

In addition to approximation, we also consider hardness in
a restricted model, namely the online model, discussed in
detail in Section 2. In the online model, the hardness is measured by
the minimum number of cuts needed to produce a proper solution. 
Lower bounds on the number of cuts needed in this model
provide a barrier for what online algorithms can achieve.  \\

Some of our ideas for finding deterministic type 2 approximation
algorithms are inspired by papers in Discrepancy Minimization, such
as \cite{AKSS}, \cite{BS2}, \cite{Bsl} and \cite{BS}. In \cite{AKSS},
the terminology refers to the \textbf{Balancer} as the entity with
the designated task of minimizing the absolute discrepancy between
agents. We adopt the same terminology here. Thus, the Balancer
has the role of an algorithm that makes cuts and assigns the resulting
intervals to agents in order to achieve a proper solution for either
$\epsilon$-Consensus Splitting or Necklace Splitting. \\

Our main algorithmic results are summarized in the
theorems below. The upper and lower bounds obtained for the online
model appear in the table at the end of this subsection.

\setcounter{section}{3}
\begin{theo}
\label{t21}
There exists an algorithm that, given $I$, an instance of
$\epsilon$-Consensus Halving for $k$ agents, and $\rho$, an oracle for
$I$ that for any $x \in [0, 1]$, $\delta \in (0, 1)$ and index $i \in
[n]$, finds $y \in (x, 1)$ so that $\mu_i([x, y]) = \delta$ (if such
a $y$ exists), makes at most $n (k - 1)$ cuts on the interval $[0,
1]$, distributing the resulting intervals to the $k$ agents so that
each one gets at least $ \frac{1}{kn}$ 
of each measure. The algorithm makes
$\text{poly}(n, k)$ oracle calls.
\end{theo}

It is worth noting that this algorithm can be implemented online
and does not require any assumption on the
behavior of the measure functions $\mu_i$. In \cite{FH}, Filos-Ratsikas
\emph{et al.} provide an efficient algorithm for the case $k = 2$ 
that gives both agents at least $\frac{1}{4}$ 
of each measure, under the
assumption that each one of the measures is uniform in some subinterval
of $[0, 1]$ and is $0$ on the rest of $[0, 1]$. The next result
deals with online algorithms. The precise online model is discussed in
the next section.

\setcounter{section}{4}
\setcounter{theo}{0}
\begin{theo}
\label{t31}
There exists an efficient, deterministic, online algorithm that, given
$I$, an instance of $\epsilon$-Consensus Halving, provides a proper
solution, making $O(\frac{n \log n}{\epsilon^2})$ cuts on the interval
$[0, 1]$.
\end{theo}

\begin{theo}
\label{t33}
There exists an efficient, deterministic, offline algorithm that,
given $I$, an instance of $\epsilon$-Consensus Halving, and $\rho$
an oracle for $I$ as in Theorem \ref{t21}, with the extra abilities to
answer queries about the sum of all the measures $\mu_i$ 
and to answer queries of the form $\mu_i([a, b])$ 
for any interval $[a, b] \subseteq [0, 1]$ and 
measure $i \in [n]$, 
provides a proper solution,
making at most $n (2 + \lceil \log_2 \frac{1}{\epsilon} \rceil)$ cuts
on the interval $[0, 1]$.
\end{theo}
\vspace{0.2cm}

\noindent
\textbf{Remark:} The offline algorithm in Theorem \ref{t33} provides a
far lower number of cuts 
than the algorithm of Theorem \ref{t31}. It is interesting to note that for
$\epsilon$ constant, this algorithm makes $O(n)$ cuts and obtains $\le
\epsilon$ absolute discrepancy while doing so with $\le n$ cuts is a
PPAD-hard problem \cite{FG2}. \\

These two algorithms for type 2 approximation for $\epsilon$-Consensus
Halving are adaptable to the Necklace Halving problem. Throughout the
paper, for Necklace Halving, we use the notation $m = \max_{i \in [n]}
m_i$ where $m_i$ is the number of beads of color $i$. 
The results below bound the number of cuts these adapted algorithms make
to reach a proper solution.

\begin{theo}
\label{t34}
There exists an efficient, deterministic, online algorithm that, given
$I$, an instance of Necklace Halving, provides a proper solution, making
at most $O(m^{2/3} \cdot n (\log n)^{1/3})$ cuts.
\end{theo}

\begin{theo}
\label{t35}
There exists an efficient, deterministic, offline algorithm that, given
$I$, an instance of Necklace Halving, provides a proper solution, making
at most $O(n \log m)$ cuts.
\end{theo}

In \cite{BL2} and \cite{BT} the authors describe offline
algorithms for an $\epsilon$ approximation version of
Necklace Halving and for $\epsilon$-Consensus Halving, making 
$O((\frac{1}{\epsilon})^{\Theta(n)})$ cuts. Our results here improve
these algorithms significantly. \\

The algorithmic results in the online model, and the nearly
matching lower bounds we establish appear
in the table below. Note
that the algorithms, for both the $\epsilon$-Consensus Halving
and Necklace Halving problems are optimal up to constant factors for any
fixed constant $n \geq 3$.
In the lower bounds for Necklace
Halving, we always assume that  $m_i = m$ for all $i \in [n]$. 

\begin{center}
 \begin{tabular}{ | m{4.0cm} | m{2.7cm} | m{3.0cm} | m{3.5cm} | } 
 \hline
 Problem & $n = 2$ measures & $n \ge 3, \enspace n = O(1)$ 
measures & $n$ measures (general case) \\ [0.5ex] 
 \hline
 Online $\epsilon$-Consensus Halving, upper bound & 
$O(\frac{1}{\epsilon^2})$ & $O(\frac{1}{\epsilon^2})$ & 
$O(\frac{n \log n}{\epsilon^2})$ \\ 
 \hline
 Online $\epsilon$-Consensus Halving, lower bound & 
$\Omega(\frac{1}{\epsilon})$ & $\Omega(\frac{1}{\epsilon^2})$ & 
$\Omega(\frac{n}{\epsilon^2})$ \\ 
 \hline
 
 Online Necklace Halving, upper bound & 
$O(m^{2/3})$ & $O(m^{2/3})$ & $O(m^{2/3} \cdot n (\log n)^{1/3} )$ \\
 \hline
 Online Necklace Halving, lower bound & 
$\Omega(\sqrt{m})$ & $\Omega(m^{2/3})$ & $\Omega(n \cdot m^{2/3})$ \\
 \hline
 
\end{tabular}
\end{center}

\setcounter{section}{1}

The structure of the rest of the paper is as follows: in 
Section 2 we describe the computational models for the offline and online
versions. In Sections 3 and 4 we present the algorithms for type 1
approximation and type 2 approximation, respectively. Subsection 4.1
contains the online algorithms, while subsection 4.2 contains the offline
algorithms. Section 5 contains the lower bounds for the online
model. In subsection 5.1 we deal with Online $\epsilon$-Consensus Halving,
and in subsection 5.2 with Online Necklace Halving. The final Section
6 contains remarks about possible extensions and open problems. To simplify
the presentation we omit all floor and ceiling signs throughout the
paper whenever these are not crucial. All logarithms are in base
$2$, unless otherwise specified.

\section{Computational models and online versions}
We first present the offline computational models, and then introduce the
online versions for both problems and their corresponding computation
models. In total, we have four models, one corresponding to each of
the combinations $\epsilon$-Consensus Splitting/Necklace Splitting and
online/offline. Each of these four models pertains to both types of
approximation. \\

The input for Necklace Splitting, for an instance with $k$ agents and $n$
colors, consists of a series of indices, each one taking a value in $[n]$,
which represents the color of the respective bead. The runtime
is, as usual,
the number of basic operations the algorithm makes to provide
a solution. \\

For $\epsilon$-Consensus Halving, we have an oracle $\rho$
that answers two types of queries. The first type
takes as input a measure index $i$, a positive quantity $\delta$,
and a starting point $x \in [0, 1]$ and returns the smallest point $y
\ge x$ so that $\mu_i([x, y]) = \delta$ if such a point exists or $1$
otherwise. The oracle can also take as input a starting point 
$x \in [0, 1]$ and positive quantity $\delta$ and
return the smallest $y \ge x$ so that
$\sum_{i=1}^n \mu_i([x, y]) = \delta$.
The second type of query takes as input
two points $0\le a < b \le 1$ 
and an index $i \in [n]$
and returns the value $\mu_i([a, b])$. 
In terms of runtime, we consider both the number of oracle
queries made and the computation done besides the queries. We seek
algorithms that are efficient in terms of both queries and computation. \\

It is worth mentioning that this computational model is not 
exactly  
the one used in \cite{FG}, but the two models are polynomially
equivalent. 
Next we discuss the online models, starting with Necklace
Splitting. The parameters 
$k$, $n$ and $m_i$ for $i \in [n]$ are given in
advance. 
The beads are revealed one by one in the following way: for  
integral $t, 0 \le
t \le \sum_{i \in [n]} m_i - 1$, at time $t$ the Balancer receives
the color of bead number $t+1$ and is given the opportunity to make a cut
between beads $t$ and $t + 1$, where this decision is
irreversible. If a cut is made, and $J$ is the newly created interval,
the Balancer also has to choose immediately the agent that gets
$J$,
before advancing to time $t + 1$. \\

The notion of time $t$ 
appears in the model for Online $\epsilon$-Consensus 
Splitting  too. Here $t$ moves from $0$ to $1$. 
Since a
continuous motion is an unpractical computational model, 
the set of possible values of $t$ where cuts are allowed
is  a discrete set of values.
This
set consists of a sequence of points $0 = x_1 < x_2 < ... < x_m = 1$,
where $m = \text{poly}(n, \frac{1}{\epsilon})$. 
At time $t = x_i$, the Balancer
receives access to the oracle $\rho$ on $[0, x_{i+1}]$ which can
provide the values of $\mu_j(x_s,x_r)$ for all $s<r \leq i+1$ and
each of the measures $\mu_j$,
and has to make
the irreversible decision of whether or not to cut at $x_i$. If he
chooses to cut, and $J$ is the newly created interval, the Balancer also
needs to decide on the spot which agent gets $J$.
After this decision is made, $t$ advances to 
$x_{i+1}$. Note that this means that cuts can only be made
at points $x_i, i \in [m]$. \\

In order to avoid pathological constructions where
there is too much measure in a small subinterval of $[0, 1]$ 
it is needed
to set an upper bound for each quantity 
$\mu_j([x_i, x_{i+1}])$. We set this upper bound
to be $\frac{\epsilon^2}{100 \log n}$ for $k=2$ and
$\frac{\epsilon^2}{100 k \log (nk)}$ for general $k$.  
Therefore, in the online
model the values of the measures  are provided on all members
of a partition of $[0,1]$ into 
polynomially many subintervals, where every measure of each
subinterval is not too large.

\section{Positive measures}
In this section, we present the proof of the result for type 1
approximation: \\

\noindent
{\bf Proof of Theorem \ref{t21}:}

The algorithm traverses the interval once and makes $kn$ marks on it,
creating $kn$ \textit{marked intervals}. Then, it chooses at most $n (k
- 1)$ of these marks for the cuts. More precisely, in the first stage,
the algorithm uses oracle calls to determine the points in $[0, 1]$
where the marks need to be made, and in the second stage determines at
which marks to make cuts.  

Let $x$ be the point on the interval up to which we have traversed
so far. Whenever we make a mark, the interval between the previous mark
and the one we just made is called a \textit{marked interval}, and it
receives a label corresponding to one of the measures, according to a rule
specified next. For each $i \in [n]$, call measure $i$ \textit{active}
if no more than $k - 1$ of the marked intervals got label $i$. If a
measure is not active at a certain point, ignore it
for the rest of the traversal. When none of the measures are
still active, stop the first stage of the algorithm. At
the beginning, $x = 0$, and all the $n$ measures are active. 

Suppose we are at a certain point $x < 1$, either the starting point
or some marked point, and there is at least one active measure. Let $y
\ge x$ be the smallest real number so that $\mu_i([x, y]) = \frac{1}{kn}$
for some $i$ that is an active measure. Mark the point $y$. The marked
interval $[x, y]$ receives label $i = \argmin_{i \enspace \text{active}}
\min \{y| \mu_i([x, y]) = \frac{1}{kn}\}$. Keep going by updating
$x$ to be $y$ until either all measures 
become inactive, or we run out of measure
and there is no $y$ that satisfies the condition above. 

We next prove that in this first phase, the algorithm does not run out
of measure, and it finishes when no measure is active, hence making $kn$
marks. Indeed, suppose that when the algorithm stops, measure $i$
is still active. Let $y_1, y_2, ..., y_p$ be the marked points until
then. Since each measure gets at most $k$ marked intervals labelled with
its index, we have that $p \le kn - 1$. Because label $i$ has always been
active, it follows that $\mu_i([0, y_1]) \le \frac{1}{kn}$ and for each
$j \in [p - 1]$, $\mu_i([y_j, y_{j + 1}]) \le \frac{1}{kn}$. This leads
to $\mu_i[0, y_p]) \le \frac{kn - 1}{kn}$, which means that $\mu_i[y_p,
1] \ge \frac{1}{kn}$, contradicting the fact that we have run out of
measure. 

Next, we present the second stage of the algorithm, which chooses which
ones of the $kn$ marks become cuts. To give each of the $k$ agents $1/kn$
of each measure, it is enough to give, for each $i \in [n]$, one of the
$k$ marked intervals of label $i$ to each agent. In other words, it 
suffices to split the labeled intervals evenly among the $k$ agents. 
This is equivalent to the Necklace Splitting Problem for $n$ colors and $k$
agents, where there are exactly $k$ beads of each color. Hence, for the
second stage of the algorithm, it is enough to prove the following:

\setcounter{theo}{1}
\begin{lemma}
\label{l21}
There exists an efficient algorithm that solves any instance of Necklace
Splitting for $n$ colors and $k$ agents where there are exactly $k$
beads of each color, making at most $n (k - 1)$ cuts.
\end{lemma}

\noindent
{\bf Proof of Lemma \ref{l21}:}

Traverse the necklace once bead by bead and cut between any pair of
consecutive beads unless the second one is the first appearance of a bead
of color $i$ for some $i \in [n]$. After each cut made, if $S$ is the set
of colors present in the newly created interval $J$, we allocate $J$
to an agent that has not received up to that point any beads of any color
in $S$. To show that after each cut such an agent exists, first note that
by the description above, 
no agent receives two beads of the same color. If
$J$ contains only one bead and its color is $i$, 
there must exist an agent
who has not received any bead of color $i$ up to that point, as there
are as many agents as beads of color $i$. If $J$ has $p \ge 2$ beads,
of colors $c_1, ..., c_p \in [n]$ appearing in this order, 
we can still give it to an agent that
has not received any bead of color $c_1$, since each of the other
beads in $J$ has a color that has not appeared before.

It thus follows that with this
allocation rule each agent
gets exactly $1$ bead of each color. To prove the upper bound on
the number of cuts, note that for each $i \in [n]$, we never cut right
before the first bead of color $i$ that appears on the necklace. Hence,
there are exactly  $n - 1$ beads (besides the very first one) 
with no cut
right before them. Since there are at most $kn - 1$ points between
consecutive beads
the algorithm makes  exactly
$kn- 1 - (n - 1) = n (k - 1)$ cuts.
\hfill $\Box$

Observe that the algorithm is polynomial in $n, k$. The first
stage makes $kn$ marks, each of which is done by finding
the smallest $y \ge x$ so that $\mu_i([x, y]) = \frac{1}{kn}$. To find
this $y$, we call the oracle $\rho(x,i, \frac{1}{kn})$, on all of the
active measures $i$, and take the smallest $y$ found. Hence, by the time
the first stage of the algorithm is over, we have made $O(n^2 \cdot k)$
oracle calls and $O(n^2 \cdot k)$ non-oracle operations. The proof
of Lemma \ref{l21}
shows that the second stage is also efficient, completing the proof.
\hfill $\Box$

\textbf{Remarks}:

\begin{itemize}

    \item The algorithm described in the proof of Lemma \ref{l21}
    traverses the necklace only once and distributes newly created
    intervals right after making cuts. Hence, this algorithm works in
    the online version for Necklace Splitting as well (for $k$ beads of
    each color).

    \item Since the algorithm that decides where to make cuts works
    in the online model, the algorithm for type 1 approximation can be
    adapted to the online model as well.

    \item Note that $\frac{1}{kn}$ tends to $0$ as $n$ tends to
infinity, that
    is, even if there are only two agents the algorithm
    does not ensure a constant amount bounded
    away from $0$ of each of the $n$ measures to each of them. 
    In the next section we describe efficient algorithms that allow
    more cuts and achieve better discrepancy.
\end{itemize}

\section{Upper bounds}

\subsection{Online algorithms}

\subsubsection{The $\epsilon$-Consensus Halving Problem}

\noindent
{\bf Proof of Theorem \ref{t31}:}

We describe an online algorithm that runs in polynomial time (in
$\frac{1}{\epsilon}$, $n$ and the input describing the 
measures) and achieves discrepancy at most $ \epsilon$. The algorithm
makes $O(\frac{n \log n}{\epsilon^2})$ cuts.
Note that by the result of 
Filos-Ratsikas and Goldberg, with only $n$ cuts, the problem of
obtaining discrepancy $\le \epsilon$, for $\epsilon$ inversely-polynomial
in $n$, is PPA-complete. By allowing more cuts we can get a
poly-time algorithm that achieves this discrepancy, and it even works
online. It is worth mentioning that this algorithm is a derandomization
of a simple randomized algorithm which cuts the interval into pieces
of small $i$-measure for all $i$ and then assigns them randomly and
uniformly to the two agents. 

Put $g=g(n,\epsilon) = \frac{\epsilon^2}{8 \log n}$. 
Traverse the interval, and whenever after the last cut made at a
point $x$  we reach a point $y$ so that
$[x, y]$ is valued at least $\frac{1}{2}
g$ by some measure and at most $g$ by all
other measures, we make a cut. Note that if we had access 
to the oracle $\rho$
on $[0, 1]$ described in 
the beginning, we could have simply set $y = \min_{i \in [n]}
\rho(x, i, g)$. However, in the online model, the Balancer
has no access to the oracle on all of $[0, 1]$, and can only
make cuts at the prescribed points
$x_i$. Hence, in this model it might happen that the
Balancer, having made the last cut at $x_p$, has to decide at time $x_j$
whether to cut or not, knowing that $\mu_i([x_p, x_j]) \leq g$
for every $i$ but that
$\mu_i([x_p, x_{j+1}]) > g$
for some $i$. When this
occurs, the Balancer simply cuts at $x_j$. By the assumption
that $\mu_i([x_j, x_{j+1}]) \le \frac{\epsilon^2}{100 \log n}$ for every
measure $i$, it follows  that in this 
case $\mu_i([x_p, x_j]) \ge \frac{\epsilon^2}{8
\log_2 n} - \frac{\epsilon^2}{100 \log n}>
\frac{1}{2} g$. 

To decide about the allocation of the interval created we define, for each
measure $i \in [n]$, a potential function $\phi_i(t)$, and a function
$\psi_i(t)$ that is an upper bound of $\phi_i$ and is computable 
efficiently.
The variable $t$ here will denote, throughout the algorithm,
the index of the last cut made. 
Define $\phi = \sum_{i = 1}^n \phi_i$ and $\psi =
\sum_{i = 1}^n \psi_i$. After each cut at step $t$, the allocation
of the interval created is chosen so as to minimize $\psi(t)$. 

The functions $\phi_i, \psi_i$ are defined by considering an 
appropriate probabilistic process.  For each $i
\in [n]$, let $X_i$ be the random variable whose value is the
difference between the $i$-th measure of the share of agent 1
and that of agent 2
if after each cut the interval created is assigned to a uniform
random agent. Let
$\epsilon_k$ be $1$ if the $k$'th interval is assigned to agent 1
and $-1$ otherwise. Therefore 
$X_i = \sum_{j = 1}^m \epsilon_j a_j$, where $m-1$ is
the total number of cuts made and $a_j = \mu_i(I_j)$, where $I_j$
is the
$j$'th created interval. The distribution defining $X_i$ is the
one where each $\epsilon_j$ is $1$ or $-1$ randomly, uniformly and
independently. 
The function $\phi_i(t)$ is defined as follows
$$
\phi_i(t) = \E{\frac{e^{\lambda X_i} + e^{- \lambda X_i}}{2} |
\epsilon_1, \epsilon_2, ..., \epsilon_t}$$

This is a conditional expectation, where the conditioning is on 
the allocation of the first $t$ intervals represented by 
$\epsilon_1, \ldots ,\epsilon_t$, and where
$\lambda=\frac{4 \log n}{\epsilon}$. (The specific choice of
$\lambda$ will become clear later). 
Since $X_i = \sum_j \epsilon_j a_j$,
where $a_j$ is the valuation of the $j$'th interval by measure $i$,
we have that $$\phi_i(t) = \E{\frac{e^{\lambda \sum_j \epsilon_j a_j}
+ e^{- \lambda \sum_j \epsilon_j a_j}}{2} | \epsilon_1, \epsilon_2,
..., \epsilon_t}$$.

The function
$\psi_i(t)$ is defined in a way ensuring 
it upper bounds the function $\phi_i(t)$. It is convenient to split each
$\phi_i(t)$ into 
$$
\frac{1}{2}\E{e^{\lambda \sum_j \epsilon_j a_j} | \epsilon_1,
.., \epsilon_t} + \frac{1}{2} \E{e^{-\lambda \sum_j \epsilon_j a_j} | \epsilon_1,
.., \epsilon_t}.
$$

For simplicity, denote the first term $\phi_i'$
and the second term $\phi_i''$. Therefore
$$ 
\phi_i'(t) 
= \frac{1}{2}\E{e^{\lambda \sum_j \epsilon_j a_j} | \epsilon_1, ..,
\epsilon_t} 
= \frac{1}{2} e^{\lambda \sum_{j = 1}^t \epsilon_j a_j} \cdot \prod_{j
\ge t + 1} (\frac{e^{\lambda a_j} + e^{- \lambda a_j}}{2}) 
$$ 
$$ 
= \frac{1}{2}e^{\lambda \sum_{j = 1}^t \epsilon_j a_j} 
\cdot \prod_{j \ge t + 1}
\cosh(\lambda a_j)
$$

A similar expression exists for $\phi_i''$. 
Define $s_t = \sum_{j = 1}^t a_j$ and $u_t = \sum_{j = 1}^t \epsilon_j
a_j$. By the discusson above
$$
\phi_i(t) =\frac{e^{\lambda u_t}+e^{-\lambda u_t}}{2}
\prod_{j \geq t+1} \cosh(\lambda a_j).
$$
Using the 
well-known inequality that $\cosh(x) \le e^{x^2 / 2}$, it follows
that 
$$
\phi_i(t) \leq \frac{e^{\lambda u_t}+e^{-\lambda u_t}}{2}
e^{\lambda^2 \sum_{j \ge t + 1} a_j^2 / 2}.
$$

By the way the cuts are produced
$a_j \le g$
for all $j$, and hence

$$
\sum_{j = t + 1} a_j^2 \le \max_{j \ge t + 1} (|a_j|) \cdot (\sum_{j
\ge t + 1} a_j) \le g \cdot (\sum_{j \ge t + 1} a_j) = g(1-s_t).
$$

Therefore
$$
\phi_i(t) \leq  \frac{e^{\lambda u_t}+e^{-\lambda u_t}}{2}
e^{\lambda^2 g (1-s_t)/2}.
$$

Define $\psi_i(t)$ to be the above upper bound for $\phi_i(t)$,
that is
$$
\psi_i(t) = \frac{e^{\lambda u_t}+e^{-\lambda u_t}}{2}
e^{\lambda^2 g (1-s_t)/2}.
$$
Note that $\psi_i(t)$  can be easily computed 
efficiently at time $t$, since 
$s_t$ and  $u_t$ (as well as $g$ and $\lambda$) 
are known at this point.

Recall that $\phi(t)=\sum_i \phi_i(t)$ and $\psi(t)=\sum_i
\psi_i(t)$. At time $t+1$, the online algorithm chooses 
$\epsilon_{t + 1}$ in order to minimize
$\psi(t+1)$. We next show that this ensures that 
$\psi(t)$ is (weakly)
decreasing in the variable $t$. To do so, it is enough to prove
that $\psi(t) \ge \frac{\psi(t+1 | \epsilon_{t+1} = 1) + \psi(t+1
| \epsilon_{t+1} = -1)}{2}$, where $\psi(t+1 | \epsilon_{t+1}
= \chi)$ denotes the value of $\psi(t+1)$ if we choose
$\epsilon_{t+1} = \chi \in \{-1, 1\}$. It 
suffices to show that for every measure
$i$, $\psi_i(t) \ge \frac{1}{2} [\psi_i(t+1 | \epsilon_{t+1}
= 1] + \frac{1}{2} [\psi_i(t+1 | \epsilon_{t+1} = -1]$. 

We proceed with the proof of this inequality. To do so, note that
$$
\psi_i(t+1 | \epsilon_{t+1}=1) =
\frac{e^{\lambda (u_t+a_{t+1}) }+e^{-\lambda (u_t+a_{t+1})}}{2}
e^{\lambda^2 g (1-s_t-a_{t+1})/2},
$$
and
$$
\psi_i(t+1 | \epsilon_{t+1}=-1) =
\frac{e^{\lambda (u_t-a_{t+1}) }+e^{-\lambda (u_t-a_{t+1})}}{2}
e^{\lambda^2 g (1-s_t-a_{t+1})/2}.
$$
Therefore
$$
\frac{\psi_i(t+1 | \epsilon_{t+1}=1)+\psi_i(t+1 |
\epsilon_{t+1}=-1)}{2} 
=
\frac{e^{\lambda u_t}+e^{-\lambda u_t}}{2} \cdot 
\frac{e^{\lambda a_{t+1}}
+e^{-\lambda a_{t+1}}}{2} 
e^{\lambda^2 g (1-s_t-a_{t+1})/2}
$$
$$
\leq 
\frac{e^{\lambda u_t}+e^{-\lambda u_t}}{2} \cdot
e^{\lambda^2 g a_{t+1}/2}
e^{\lambda^2 g (1-s_t-a_{t+1})/2}
=
\frac{e^{\lambda u_t}+e^{-\lambda u_t}}{2} \cdot
e^{\lambda^2 g (1-s_t)/2}=\psi_i(t),
$$
as needed.

The process of selecting $\epsilon_{t+1}$ so that
$\psi(t+1)$ is minimized is performed for every a new cut. Note that the
total number of cuts is bounded by $\frac{n}{\frac{1}{2} g(n, \epsilon)}
= 16 \frac{1}{\epsilon^2} n \log n$. Hence, this process takes polynomial
time in the quantities mentioned, and it splits the interval among the
two agents. 

Next, we prove that the absolute discrepancy is smaller than
$\epsilon$. Let $m$ be the number of cuts made.
Note that at $t=m$,
$\phi_i(m)=\psi_i(m)$ 
is not an expectation, but a realization of the random process
determined by following the algorithm. At the end
$\psi_i(m) = \frac{e^{\lambda x_i} + e^{- \lambda x_i}}{2}$,
where $x_i$ is the difference between the amount of measure $i$ of agent
1 and the amount of measure $i$ of agent 2 at the end of the
algorithm. If for some measure $i$ the absolute value of the
discrepancy is $ > \epsilon$, then $\phi_i(m) > \frac{e^{\lambda
\epsilon}}{2}$, which means that $\psi(m) \geq \phi(m) 
\geq \frac{e^{\lambda
\epsilon}}{2}$. Since $\psi(t)$ is decreasing in $t$, it is enough
to prove that $\psi(0) \le \frac{e^{\lambda \epsilon}}{2}$. Note that
$\psi(0) \leq n e^{\frac{1}{2} \lambda^2 g}$, hence it suffices to
prove that $n e^{\frac{1}{2} \lambda^2 g}
< \frac{e^{\lambda \epsilon}}{2}$, which is equivalent to $\frac{1}{2}
\lambda^2 g + \log_e (2n) < \lambda \epsilon$. Recall that
$\lambda = \frac{4 \log n}{\epsilon}$. The
inequality becomes $\log n + \log_e (2n) < 4 \log n$, which is
clearly true. Therefore, 
the algorithm obtains discrepancy at most $\epsilon$, it 
is poly-time, online,
and makes $O(\frac{n \log n}{\epsilon^2})$ cuts. This completes the
proof.

\hfill $\Box$

\subsubsection{Necklace Halving}
In this subsection we use the previous algorithm to prove 
Theorem \ref{t34} about online Necklace Halving. Note first that
if, say, $\log n>m/1000$ 
the result is trivial, as less than $nm$ cuts  suffice
to split the necklace into single beads, hence we may and will
assume that $m \geq 1000 \log n$. \\

\noindent
{\bf Proof of Theorem \ref{t34}:}
Given a necklace with $m_i$ beads of color $i$ for $1 \leq i \leq
n$, where $m=\max m_i$, 
construct an instance of $\epsilon$-Consensus
Halving as follows. Replace each bead of color $i$ by an interval
of $i$-measure $1/m_i$ and $j$-measure $0$ for all $j\neq i$. These
intervals are placed next to each other according to the order in
the necklace, and their lengths are chosen so that altogether they
cover $[0,1]$. During the online algorithm, the beads of
the necklace  are revealed one by one. Throughout the algorithm 
we call the beads that have not yet been revealed the 
{\em remaining beads}.

Without trying to optimize the absolute constants,
define $\epsilon=10 (\frac{\log n}{m}) ^{1/3}~$ ($\leq 1$). Our algorithm
follows the one used in the proof of Theorem \ref{t31}, but when
the number of remaining beads of color $i$ becomes smaller than
$20 m^{2/3} (\log n)^{1/3}$ the algorithm makes a cut before and
after each arriving bead of color $i$, allocating it to the agent
with a smaller number of beads of this type, where ties are broken
arbitrarily.  During the algorithm we call a color {\em critical}
if the number of remaining beads of this color is smaller than
$20 m^{2/3} (\log n)^{1/3}$, otherwise it is {\em normal}. Although
the input is now considered as the interval $[0,1]$ with $n$
continuous measures on it, we allow only cuts between intervals
corresponding to consecutive beads, and do not allow any cuts in
the interiors of intervals corresponding to beads. Note that if
$\frac{1}{m_i} \leq \frac{\epsilon^2}{100 \log n}$, that is,
$m_i \geq m^{2/3} (\log n)^{1/3}$, this is consistent with our
definition of the online model for the $\epsilon$-Consensus Halving
Problem. Otherwise, color $i$ is critical from the beginning, and
in this case we cut before and after every bead of color $i$.

Starting at $t=0$, 
as in the proof of Theorem \ref{t31}, define, for each color $i$, the
upper bound functions for the potential, $\psi_i$, and put $\psi =
\sum \psi_i$, where the sum at every point is only over the
normal colors $i$. If the last cut is at point $x \in [0, 1]$,
the next cut is made before the last bead 
whose corresponding position $y$ in
the $\epsilon$-Consensus Halving instance has the property 
that $\mu_i([x,
y]) \le g(n, \epsilon)$ for all $i$. The newly created interval is then
allocated on the
spot according to the choice
that minimizes $\psi$. 

This is done until some color $i$ becomes
critical. Note that until this stage, since $\psi(t)$ is (weakly)
decreasing, the absolute discrepancy does not exceed
$\epsilon$, implying that
the discrepancy in terms of beads on the necklace, for each
color $i \in [n]$, does not exceed $10 m^{2/3} (\log n)^{1/3}$. When
$i$ becomes critical, it stops contributing to the potential
function (which as a result, becomes smaller). From this time on
color $i$ is handled separately, 
the algorithm makes a cut before
and after any occurrence of it and allocates it to the agent with a
smaller number of beads of this color. As before, the 
newly created intervals of the beads of the other colors
are allocated in order to minimize the potential $\psi$.

It is clear that the potential $\psi(t)$ here is a decreasing
function of $t$, as in the previous proof. Therefore whenever a
color becomes critical the discrepancy in it in terms of beads
is at most $10 m^{2/3} (\log n)^{1/3}$ and as the number of
remaining beads of this color is larger, the process will end
with a balanced partition of the beads of each color $i$ between the
agents, allocating to each of them either $\lfloor m_i/2 \rfloor$
or $\lceil m_i/2 \rceil$ of these beads.

To bound the number of cuts the algorithm makes call a cut {\em
forced} if it is made before or after a bead of color $i$ when $i$
is critical. The number of non-forced cuts is clearly at most
$O(\frac{n \log n}{\epsilon^2}) = O(n (\log n)^{1/3}
\cdot m^{2/3})$. The number of forced cuts cannot exceed
$2 n \cdot 20 m^{2/3} \cdot (\log n)^{1/3}$. Hence,
the total number of cuts made is $O(n (\log n)^{1/3} \cdot m^{2/3})$. 

Since the algorithm clearly works online this completes the proof
of the theorem.

\hfill $\Box$

\subsection{Offline algorithms}

\subsubsection{The $\epsilon$-Consensus Halving Problem}

\noindent
{\bf Proof of Theorem \ref{t33}:}
Given $n$ non-atomic measures $\mu_i$ on the interval $[0,1]$ we
describe an efficient algorithm that cuts the interval 
in at most $n (2 + \lceil \log_2 \frac{1}{\epsilon} \rceil)$ places
and splits the resulting intervals into two collections $C_0, C_1$ so
that $\mu_i(C_j) \in [\frac{1}{2} - \frac{\epsilon}{2}, \frac{1}{2} +
\frac{\epsilon}{2}]$ for all $i \in [n], 0 \le j \le 1$.  Note, first,
that if the collection $C_1$ has the right amount according to each of
the measures $\mu_i$, so does the collection $C_0$, hence it is convenient
to only keep track of the intervals assigned to $C_1$. For each interval
$I \subset [0,1]$ denote $\mu(I)=\mu_1(I)+ \ldots + \mu_n(I)$. Thus
$\mu([0,1])=n$. Using $2n-1$ cuts  split $[0,1]$ into $2n$ intervals
$I_1,I_2, \ldots ,I_{2n}$ so that $\mu(I_r)=1/2$ for all $r$.  For each
interval $I_r$ let $v_r$ denote the $n$-dimensional vector $(\mu_1(I_r),
\mu_2(I_r), \ldots ,\mu_n(I_r))$. 

By a simple linear algebra argument, which is a standard fact about the
properties of basic solutions for Linear Programming problems,
one can
write the vector $(1/2, 1/2, \ldots ,1/2)$ as a linear combination
of the vectors $v_r$ with coefficients in $[0,1]$, 
where at most  $n$ of them are not
in $\{0,1\}$. For completeness, we include the proof, which also
shows that one can find coefficients as above efficiently.
Start with all coefficients being $1/2$. Call a coefficient which is 
not in $\{0, 1\}$ \textit{floating} and one in $\{0, 1\}$ \textit{fixed}.
Thus at the beginning all $2n$ coefficients
are floating. As long as there are more than $n$ floating coefficients,
find a nontrivial linear dependence among the corresponding vectors
and subtract a scalar multiple of it which keeps all floating coefficients
in the closed interval $[0,1]$ 
shifting at least one of them to the boundary
$\{0,1\}$, thus fixing it. 

This process clearly ends with at most
$n$ floating coefficients.
The intervals with fixed coefficients with value $1$ 
are now assigned to the collection $C_1$ and
those with coefficient $0$ to $C_0$. 
The rest of the intervals remain.
Split each of the remaining intervals into two intervals, each with
$\mu$-value $1/4$. 
We get a collection $J_1, J_2, \ldots ,J_m$ of $m \leq 2n$
intervals, each of them has the coefficient it inherits from its original
interval. Each such interval defines an $n$-vector as before, and the sum
of these vectors with the corresponding coefficients (in $(0,1)$) is
exactly what the collection $C_1$ should still get to have its total
vector of measures being $(1/2, \ldots ,1/2)$. 

As before, we can shift the coefficients until
at most $n$ of them are floating, assign the intervals with $\{0,1\}$ 
coefficients to the collections $C_0,C_1$ and keep at most $n$ intervals
with floating coefficients. Split each of those into two intervals of
$\mu$-value $1/8$ each and proceed as before, until we get at most $n$ 
intervals with floating coefficients, where the $\mu$-value 
of each of them is at most $\epsilon / 2$. This happens after 
at most 
$\lceil 1 + \log_2 (1/\epsilon) \rceil $ rounds. In the first one, we have
made $2n-1$ cuts and in each additional round at most $n$ cuts. 
Thus the total number of cuts is at most
$n(2+ \lceil \log_2 (1/\epsilon) \rceil)-1$. 

From now on we add no additional cuts, and show how to
allocate the remaining intervals to $C_0,C_1$.
Let $\II$ denote the collection of
intervals with floating coefficients. Then $|\II| \leq n$ and
$\mu(I)\leq \epsilon / 2$ for each $I \in \II$. This means that 
$$
\sum_{i=1}^n \sum_{I \in \II} \mu_i(I) \leq n \epsilon / 2.
$$
It follows that there is at least one measures $\mu_i$ so that
$$
\sum_{I \in \II} \mu_i(I)  \leq \epsilon / 2.
$$
We can think of the remaining floating coefficients as
the fraction of each corresponding interval that
agent 1 owns. 
Observe that for any
assignment of the intervals $I \in \II$ to the two collections
$C_0,C_1$, the total $\mu_i$ measure of $C_1$ (and hence also of
$C_0$) lies
in $[1/2- \epsilon/2, 1/2+ \epsilon/2]$, as this measure with the floating
coefficients is exactly $1/2$ and any allocation of the intervals
with the floating coefficients changes this value by at most
$\epsilon / 2$. We can thus ignore this measure, for  ease of
notation assume it is measure number $n$, and
replace each measure vector of the members in $\II$ by a vector
of length $n-1$ corresponding to the other $n-1$ measures. 
If $|\II|>n-1$ (that is, if $|\II|=n$), 
then it is now possible to shift the floating coefficients
as before until at least one of them reaches the boundary, fix it
assigning its interval to $C_1$ or $C_0$ as needed, and
omit the corresponding interval from $\II$ ensuring its size is
at most $n-1$. This means that for the modified $\II$ the sum
$$
\sum_{i=1}^{n-1} \sum_{I \in \II} \mu_i(I) \leq (n-1) \epsilon / 2.
$$
Hence there is again a measure $i$, $1 \leq i \leq n-1$ so that
$$
\sum_{I \in \II} \mu_i(I)  \leq \epsilon / 2.
$$
Again, we may assume that $i=n-1$, observe that measure 
$n-1$ will stay in its desired range for any future allocation 
of the remaining intervals, and replace the measure vectors by ones
of length $n-2$. This process ends with an allocation of all
intervals to $C_1$ and $C_0$, ensuring that at the end
$\mu_i(C_j) \in [1/2-\epsilon/2,1/2+\epsilon/2]$ for all $1 \leq i \leq n$,
$0 \leq j \leq 1$. These are the desired collections. It is clear
that the procedure for generating them is efficient, requiring
only basic linear algebra operations. This completes the proof of the
theorem.  \hfill  $\Box$
\vspace{0.2cm}

\noindent
{\bf Remark:}\, 
The above Theorem shows that $3n$ cuts suffice for $\epsilon
=1/2$ (which ensures both agents get at least $1/4$ of each measure). A simple different choice of parameters implies that $(2 +
\delta)n$ cuts suffice in order to ensure that each of the two collections
$C_i$ satisfies $\mu_i(C_j) \geq \frac{\delta}{4 + 2 \delta}$ for all
$1 \leq i \leq n$ and $j \in \{0,1\}$.
\vspace{0.2cm}

\noindent
{\bf Remark:}\, 
The argument can be extended to splitting into $k$ nearly fair collections
of intervals. One way to do it is to generate the collections one by
one. See section 6 for more details.

\subsubsection{Necklace Splitting}
In this subsection, we present the Necklace Splitting algorithm obtained
by adapting the algorithm in the proof of Theorem \ref{t33}. \\

\noindent
{\bf Proof of Theorem \ref{t35}:}

Convert the given necklace into an instance of
$\epsilon$-Consensus Halving as done in the proof of Theorem
\ref{t34} and
mark the places of the cuts made by
the algorithm 
in the proof of Theorem \ref{t33} applied to the resulting
input with $\epsilon = \frac{1}{2m}$. The intervals separated by
the marks are partitioned by the algorithm 
into two collections forming a solution of the continuous problem.
Note that the continuous solution would give 
discrepancy at most
$\max_{i \in [n]} m_i \cdot \epsilon \le 1/2$ in terms of beads
if we were allowed to cut
at the marked points. 
The only subtle point is that some of the marks may be in the interior
of small intervals corresponding to beads, and we wish to
cut only between beads.  

Call a mark
between two consecutive beads \textit{fixed} and call the other
marks
\textit{floating}. We first show how to shift each of the floating marks
so that the absolute discrepancy 
does not increase beyond $1/2$ and all but at most one mark for each
color are made between
two consecutive beads. To do so, if there exists a floating mark
between two intervals assigned to the same agent
eliminate it and
merge the two intervals.
If there is no such mark and there are at least two
floating marks in the interior of intervals corresponding to color
$i$, we shift both of them by the same amount in the appropriate way 
until at least one of them becomes fixed. If during
this simultaneous shift one of the two marks arrives in a spot occupied
by a different mark, we stop the shift and discard one of the duplicate
marks. Note that the quantities the two agents receive do not change. 

This procedure reduces the number of floating marks 
until there is at most one floating mark for each color.
If there is such a floating mark, round it to the closest boundary
between beads noting that this can increase the absolute
discrepancy by at most $1$.
Therefore, once all marks are fixed, the absolute 
discrepancy is $\le 3/2$. Since all the cuts 
are between consecutive beads, this discrepancy 
has to be an integer, and thus it
is at most $1$, as desired. The number
of cuts made is $\le n (2 + \lceil \log_2 \frac{1}{\epsilon}
\rceil) =
n (3 + \lceil \log_2 m \rceil) = O(n \log m)$.

\hfill $\Box$

\section{Lower bounds}
In this section we present the lower bounds  
for $\epsilon$-Consensus 
Halving and Necklace Halving in the online model.

\subsection{The $\epsilon$-Consensus Halving Problem}

\subsubsection{Punishment and its application}
In this subsection we prove a $\Omega(\frac{1}{\epsilon})$ lower bound
on the number of cuts the Balancer needs to make in the online
model in order to
obtain a proper solution for $n = 2$. The proof relies on the idea of
punishing the Balancer if he allows too much discrepancy on one
measure at
any point. The next claim is the crux of the argument. It shows
that if we keep one measure unused as a punishing threat,
we cannot have, after any cut made, 
discrepancy exceeding $2 \epsilon$ on any other measure. Indeed, otherwise,
there exists a punishment that prevents the Balancer from achieving
$\epsilon$ discrepancy in the end with any finite number of cuts.
Recall that the  discrepancy on measure $i$ during the algorithm is
the difference between the $i$-th measure of the share of agent 1
and that of agent 2. 
\begin{lemma}
\label{l41}

Denote by $x$ the discrepancy on measure $1$
after the last cut made, and assume that measure $2$ of the shares
allocated  so far is $0$.
If $|x| \ge 2 \epsilon$, then there exists an adversarial input so
that no finite number of cuts can achieve an absolute discrepancy 
at most $\epsilon$ at the end.

\end{lemma}

\noindent
{\bf Proof of Lemma \ref{l41}:}

Assume, without loss of generality,  
that $x$ is positive, thus $x \ge 2 \epsilon$. Let $t \in (0,
1)$ be the point of the last cut made in $I = [0, 1]$. 
Distribute the rest of measures
$1, 2$ uniformly on $(t,1]$.
Suppose that a finite number of cuts, at points $t =
t_0 < t_1 < t_2 < ... < t_k < 1=t_{k+1}$ is made and the resulting
intervals are allocated to the agents with maximum discrepancy
below $\epsilon$. Denote $J_i = [t_i, t_{i+1}]$ and let $\epsilon_i$
be the sign attributed to each interval, defined to be 
$+$ if it is given to agent $1$
and $-$ if given to agent $2$. Since the discrepancy on measure $2$
is below $\epsilon$ in absolute value, it follows that
$|\frac{1}{1-t} \sum_{i
= 0}^k \epsilon_i l(J_i)| < \epsilon$, 
where $l(J_i)=t_{i+1}-t_i$ is the length
of $J_i$. However, the condition on measure $1$ yields
the inequality $x + \frac{\mu}{t-1} \sum_{i = 0}^k 
\epsilon_i l(J_i) < \epsilon$,
where $\mu$ is the remaining quantity of measure $1$ after 
the cut at $t$. Since $\mu < 1$, this leads to 
$$
\epsilon> x + \frac{\mu}{t-1} \sum_{i = 0}^k \epsilon_i l(J_i)
> x - \mu \epsilon > \epsilon,
$$
contradiction.

\hfill $\Box$

A simple application of
this lemma implies a $\Omega(\frac{1}{\epsilon})$ lower bound
for $n = 2$ measures.

\begin{theo}
\label{t41} 
There exists an adversarial input that forces any deterministic algorithm
for Online $\epsilon$-Consensus Halving with $n = 2$ measures to make
$\Omega(\frac{1}{\epsilon})$ cuts in order 
to obtain a proper solution.
\end{theo}

\noindent
{\bf Proof of Theorem \ref{t41}:}
Keep measure $2$ reserved for a potential punishment if the discrepancy
on measure $1$ ever exceeds $2 \epsilon$ in absolute value. At the
beginning, and after each cut, as long as $|x| < 2 \epsilon$, set
$\mu_1$ to be uniform with density $1$ on the 
portion that follows, until
the next cut made. If the Balancer waits for length larger than 
$4 \epsilon$
before cutting, then when the next cut is made and the resulting
interval allocated, the new discrepancy on
measure $1$ is at least  $2 \epsilon$. By the previous lemma
in this case the adversary can ensure that the 
Balancer will not be able to obtain the desired bound for the maximum
discrepancy. Hence, the distance between any two 
consecutive cuts is less
than $4 \epsilon$, yielding the required 
$\Omega(\frac{1}{\epsilon})$ lower bound.

\hfill $\Box$

Note that the upper bound provided by our online  algorithm for
$n=2$ measures is larger by a factor 
of $\Theta(\frac{1}{\epsilon})$ than this lower bound.
We next show that with one additional measure we can obtain a tight
lower bound, up to a constant factor.

\subsubsection{The case $n \ge 3$}
We first prove that for $n = 3$ measures we have a tight lower bound up to
a constant factor. For $n>3$ measures this will imply a lower bound
which is tight up to a 
$\Theta(\log n)$ factor.
\begin{theo}
\label{t42} 
There exists an adversarial input that forces any deterministic
algorithm for Online $\epsilon$-Consensus Halving for $n = 3$
measures to make
$\Omega(\frac{1}{\epsilon^2})$ cuts.
\end{theo}

\noindent
{\bf Proof of Theorem \ref{t42}:}
The proof applies the punishing argument given in Lemma
\ref{l41}. Measure number $3$ will be kept for possible 
punishment. We start with some notation and
definitions. After each cut at point $t$, let
$x_t, y_t$ denote the discrepancies (positive or negative) for measures
$1$ and $2$ respectively. The state after each such cut is 
represented by the
two dimensional vector $v_t = (x_t, y_t)$. After a new cut is
made and a new interval $J$ is formed, we view the interval as a two
dimensional vector $p = (p_1, p_2)$, where $p_i = \mu_i(J)$.
By
Lemma \ref{l41}, we may and will assume that $v_t$ is 
in the square
$[-2 \epsilon , 2 \epsilon] \times [-2 \epsilon , 2 \epsilon]$ after each
cut. The adversary tries to reveal online an input that forces the
Balancer to make many cuts to ensure that after each cut $v$ lies
in this square.
In order to analyze the progress we maintain
during the algorithm a potential function
$M(x, y)=M(v)$, which is defined in what follows.

After each cut and interval
allocation made by the Balancer, the adversarial input will consist
of measures distributed according to the
proportions $\gamma=\frac{p_1}{p_1 +
p_2}$ and $1-\gamma=\frac{p_2}{p_1 + p_2}$. 
The choice  of $\gamma$ will be made in order to ensure that
both $M(v + p) - M(v)$
and $M(v-p) - M(v)$ are large for any future
large interval. Note that if there is not enough measure of type
$1$ or $2$ left, it may be needed to limit the length of the
interval in which the measures will be distributed according to
the above proportions. 
It is convenient to define, after any cut made at point $t$, 
the \textit{feasible prefix} as the interval $[0, \ell]$, with $\ell$
being the maximum real so that the adversary can still distribute
the measures in $(t,\ell]$ according to the required proportions
without running out of measure.

The potential function is defined by
$$
M(x, y) = x^2 + y^2 + 5 \epsilon x - 5 \epsilon y.
$$ 

Let $p = (p_1, p_2)$ be the vector corresponding to the new
cut made, following the cut made at time $t$, where
$v_t=(x,y)$. Put $\alpha =
p_1 + p_2$. Assuming that $(p_1, p_2)$ are proportional to
$(10 \epsilon - 4y, 10 \epsilon + 4x)$, which is a vector with
positive coordinates as $|x|, |y| <
2 \epsilon$, we get that $$M(v + p) - M(v) = p_1^2 +
p_2^2 + 2x p_1 + 2 y p_2 + 5 \epsilon p_1 - 5 \epsilon p_2 = p_1^2 + p_2^2
+ \frac{1}{2} [p_1 \cdot (10 \epsilon + 4x) - p_2 \cdot (10 \epsilon -
4y)] = $$ $$= p_1^2 + p_2^2 \ge \frac{1}{2} \alpha^2$$

and similarly, 

$$M(v - p) - M(v) = p_1^2 + p_2^2 + \frac{1}{2} [-
p_1 \cdot (10 \epsilon + 4x) + p_2 \cdot (10 \epsilon - 4y)] = p_1^2 +
p_2^2 \ge \frac{1}{2} \alpha^2$$

With this in mind, starting at the point $t$ of the last cut, 
until the next cut is made,
put $\gamma = \frac{10 \epsilon - 4y_t}{20 \epsilon + 4(x_t-y_t)}$, and 
define the measures $\mu_1, \mu_2$ by $\mu_1([t, x]) = 4 \gamma (x-t),
\mu_2([t, x]) = 4 (1-\gamma) (x-t)$ for every $x \ge t$ on the
feasible
prefix. If the next cut made by the Balancer is made at $t' > t$
(in the feasible prefix), for any allocation of the interval
obtained
we get $M(v_{t'}) - M(v_t) \ge
\frac{1}{2} \alpha^2$, where $\alpha = 4 (t' - t)$. Note that
after each cut made
the adversary modifies $\mu_1, \mu_2$ in the following 
part of the interval according to the rule above.

Let $v_f$ be the vector after the last cut made in the feasible
prefix,
and denote $M_f = M(v_f)$. Note that $|x_f|, |y_f| < 2 \epsilon$,
and thus
$|y_f-x_f| < 4 \epsilon$. Hence, $M_f \le 28 \epsilon^2$. Assume that
the Balancer makes $r$ cuts in the feasible prefix 
(note that by Lemma \ref{l41} he can never allow more than $4
\epsilon$ measure of any type to arrive without a cut).
By Cauchy-Schwartz, since the total measure of the
feasible prefix is at least $1+4 \epsilon$,
we have that the total increase in the function $M$ since its
initial value $0$ is
at least $\frac{1}{2} (\frac{1}{r})^2 \cdot r = \frac{1}{2r}$.
Therefore, 
$\frac{1}{2r} \le 28 \epsilon^2$, which yields $r \ge \frac{1}{56
\epsilon^2}$, showing that $r = \Omega(\frac{1}{\epsilon^2})$, 
as desired.

\hfill $\Box$

If the number $n$ of measures is larger than $3$
it is possible to divide the interval into $\lfloor n/3 \rfloor$ 
subintervals, using $3$ of the measures as above in each of them  and
getting a lower bound of
$\Omega(\frac{n}{\epsilon^2})$ for the total number of cuts.
Recall that our online algorithm performs 
$O(\frac{n \log n}{\epsilon^2})$ cuts, matching this lower bound up
to a logarithmic factor.

\subsection{Necklace Halving}
As in subsection 5.1, we first provide a preliminary lower bound for $n =
2$ colors, using one of the colors as a punishing threat.

\subsubsection{A preliminary bound}
We provide a $\Omega(\sqrt{m})$ lower bound for the number of cuts
required in any online algorithm
when the number of colors is $n=2$ and there are $m$ beads of
each color. The argument is similar to the one for the
$\epsilon$-Consensus Halving Problem, but since each bead can have
only one of the colors it is impossible to distribute the two colors
evenly in an interval. We thus need the following simple lemma,
which is a special case of a more general elegant result of Tijdeman
\cite{Ti}. Since this special case is much simpler, we include its
proof, for completeness.
\begin{lemma}
\label{l91}
For every real $\gamma \in [0,1]$ there is an infinite 
binary sequence $a_1,a_2,a_3, \ldots $ so that in every prefix of it
$a_1, a_2, \ldots , a_j$  the number of elements $a_i$ which are
$1$ deviates from $\gamma j$ by less than $1$.
\end{lemma}
\begin{proof}
By compactness it suffices to prove the existence of such a
sequence of any finite length $r$. Consider the following system of
linear inequalities in the variables $x_1, x_2, \ldots ,x_r$:
$0  \leq x_i \leq 1$ for all $1 \leq i \leq r$, and for every $j
\leq r$,
$\lfloor \gamma j \rfloor \leq \gamma j \leq \lceil \gamma j
\rceil$. This system has a real solution $x_i=\gamma$ for every $i$
and the matrix of coefficients of the constraints is totally
unimodular. Hence there is an integral solution $x_i=a_i \in
\{0,1\}$ providing the required sequence.
\end{proof}

We use the following notation.
During the algorithm let
$t$ denote the number of beads revealed so far. If a cut is made
at this point, let
$x_t$ be the difference between the number of beads of
color 1 allocated  to agent 1 and the number of
beads of color 1 allocated to agent 2. Define $y_t$
similarly for beads of color 2. Let $\alpha_t, \beta_t$ denote
the number of remaining  beads of colors $1$ and $2$,
respectively.

\begin{lemma}
\label{l42}

Let $\Delta$ be a positive integer.  Suppose that a cut is made at
point $t$ and 
$|x_t| = \Delta$ and assume that no bead of color $2$ appeared so
far. Then there exists an adversarial
input that forces the Balancer to make at least 
$\Delta / 4 = \Omega(\Delta)$ cuts.

\end{lemma}

\noindent
{\bf Proof of Lemma \ref{l42}:}

Without loss of generality assume that
$x_t = \Delta >0$. Note that  by assumption $\beta_t=m$ and
$\alpha_t <m$. Put
$\gamma =\frac{m}{\alpha_t +m}$ and note that $\gamma>1/2$. 
By Lemma \ref{l91} it is possible to choose an ordering of the
remaining $\alpha_t+m$ beads of the necklace  so that in every 
prefix of it of any length $j$, the number of beads of color
$2$ deviates from $\gamma j$ by less than  $1$. Since our online
model allows the Balancer to see the next bead before the decision
to make a cut preceding it  we may have to change the first bead 
in this ordering, this still ensures that in any interval of
length $\ell$ in the remainder of the necklace, the number of 
beads of color $2$ deviates from $\gamma \ell$ by at most $2$. 

Suppose the Balancer cuts the remainder of the necklace
and allocates the resulting intervals $R_1, ..., R_u$ to
agent 1 and $T_1, ..., T_v$ to agent 2 to obtain a balanced
allocation. For each one of these intervals $I$ let $\ell(I)$
denote its
length. By assumption
at time $t$ agent 1 has exactly $\Delta$ more beads than agent 2.
Since at the end each agent has half of the beads (for simplicity
we assume that $m$ is even),
$\sum_{i = 1}^v \ell(T_i) -
\sum_{j = 1}^u \ell(R_j) = \Delta$. 

By construction, the total number of beads of color $2$ in all
intervals $T_i$ deviates from $\gamma \sum_{i = 1}^v \ell(T_i)$
by at most $2 v$. Similarly, the total number of beads of color $2$
is all intervals $R_j$ deviates from $\gamma \sum_{j = 1}^u
\ell(R_j)$ by at most $2u$. As these two numbers should be equal
it follows that 
$$
\gamma \Delta =\gamma (\sum_{i = 1}^v \ell(T_i) -
\sum_{j = 1}^u \ell(R_j)) \leq 2u+2v
$$
This implies that $2(u+v) \geq \gamma \Delta > \Delta/2$ 
and as the number of cuts is
at least $u+v$ the desired result follows.

\hfill $\Box$

The last lemma easily implies the following.

\begin{theo}
\label{t43} 
There exists an adversarial input that forces any deterministic algorithm
for Online Necklace Halving with $n = 2$ colors to make $\Omega(\sqrt{m})$
cuts in order to obtain a proper solution.
\end{theo}

\noindent
{\bf Proof of Theorem \ref{t43}:}

Put $\Delta = \sqrt{m}$ and proceed by revealing only beads of color $1$. 
By Lemma \ref{l42}, if after  
a cut  at some $t$, $|x_t| > \sqrt{m}$ the desired result follows.
Otherwise it is clear 
the number of beads between any two consecutive
cuts 
is less than $2\sqrt{m}$, implying that the
total number
of cuts made by the Balancer is $\Omega(\sqrt{m})$.
\hfill $\Box$

\subsubsection{A nearly tight bound}
\begin{theo}
\label{t44} 
An adversary can force any deterministic algorithm
for Online Necklace Halving with $n = 3$ colors and $m$ beads of
each color to make $\Omega(m^{2/3})$ cuts.
\end{theo}

\noindent
{\bf Proof of Theorem \ref{t44}:}
As in the previous subsection, let $x_t$ denote the
discrepancy between the number of beads of color $1$  
allocated to agent $1$
and that allocated to agent $2$ after cut $t$, and let $y_t$ denote
the corresponding discrepancy for color $2$, where color $3$ will
be kept as a  punishment threat. We proceed by revealing only beads
of the first two colors.
By Lemma \ref{l42} with $\Delta = m^{2/3}$ the
Balancer needs to maintain $|x_t|, |y_t| \le m^{2/3}$, since otherwise
the adversary can force $\Omega(m^{2/3})$ cuts, using beads of the
third color. Hence we assume that
during the process of revealing the initial $m+4m^{2/3}$ beads of the
necklace $x_t,y_t$ stay in the above range after each cut.

Define a potential function 
$$M(x, y) = x^2 + y^2 + 5 m^{2/3} (x - y).$$ 

After a cut with $v_t=(x_t,y_t)=(x,y)$ define
$\gamma = \frac{10 m^{2/3} - 4y}{20 m^{2/3} + 4 (x - y)}$. Note
that $0 < \gamma < 1$, as $|x|, |y| \le m^{2/3}$. By Lemma 
\ref{l91}  it is possible to order the remaining part of the first 
$m+4m^{2/3}$  beads of the necklace  so that in each prefix of any length
$j$ of this remaining
part the number of beads of color $1$ deviates from $\gamma j$ by
less than $1$ and the number of beads of color $2$ deviates by less
than $1$ from $(1-\gamma)j$. As the first bead of this remaining
part has been observed already by the Balancer we may need to
change one bead in this ordering, getting a deviation of less than 
$2$ in each prefix. This means that if the next cut will be made
after some $j$ additional beads, the vector $p=(p_1,p_2)$ of
additional beads of colors $1$ and $2$, respectively, can be
written as a sum of the vector $p'=(\gamma j, (1-\gamma)j)$ and an
error vector $\delta=(\delta_1,\delta_2)$ 
of $\ell_{\infty}$-norm smaller
than $2$. By a simple computation analogous to the one done
in the proof of the lower bound for the $\epsilon$-Consensus Halving
problem it follows that
$M(v_t+p')-M(p') \geq \frac{j^2}{2}$
and 
$M(v_t-p')-M(p') \geq \frac{j^2}{2}$.
A simple computation using the fact that $|x|,|y| \leq m^{2/3}$
and that a similar bound holds after adding or subtracting the vector 
$p'$ shows that adding or subtracting the vector
$\delta$ can decrease the value of $M$ by less than
$15 m^{2/3}$. Therefore the value of $M$ increases by at least
$j^2/2-15 m^{2/3}$ with a cut of $j$ beads. 

Suppose that we have $r$ cuts among the first $m+4m^{2/3}$ beads
of the necklace, and the lengths of the resulting intervals are
$j_1, j_2, \ldots ,j_r$. Since throughout the process 
$|x_t|, |y_t| \leq m^{2/3}$, it follows that
$M(x_t,y_t) \leq 12m^{2/3}$. On the other hand  by the above
discussion the value of $M$ at the end is at least
$\sum_{i=1}^r \frac{j_i^2}{2}-15 m^{2/3} r$.
Since $\sum_{i=1}^r j_i \geq m$ (as we cannot have $4m^{2/3}$
consecutive beads with no cut among them), it follows, by
Cauchy-Schwartz, that $\sum j_i^2 \geq \frac{m^2}{r}
$. This implies that
$$
\frac{1}{2}\frac{m^2}{r} 
- 15 r m^{2/3} \leq 12 m^{4/3}
$$
showing that $r = \Omega(m^{2/3})$, as needed. 

\hfill $\Box$

\noindent
{\bf Remark:}\, 
For $n>3$ colors with $m$ beads of each color
one can consider a necklace consisting of 
$\lfloor n/3 \rfloor$ segments with at least $3$ colors in each of
them. The above argument shows that it is possible to force
$\Omega(m^{2/3})$ cuts in each segment, implying
an $\Omega(n m^{2/3})$ lower bound. Thus, for $n$ colors,
the gap between our lower and upper bounds  for the number of cuts
required is only a factor of $\Theta((\log n)^{1/3})$.

\section{Extensions and open problems}
We conclude with some generalizations of the algorithms presented and
the lower bounds obtained, and with comments on the questions that remain
open. For the generalizations, we include the statements and a brief
overview of each of the proofs.

\subsection{Generalizations}
Our online algorithm for $\epsilon$-Consensus Halving can be adapted for
the general case of $k$ agents for $\epsilon$-Consensus Splitting so as to
provide a proper solution making $O(\frac{k n \log (nk)}{\epsilon^2})$
cuts. To obtain a proper solution for the general case of $k$ agents,
it suffices to make the absolute discrepancy at most $\epsilon/k$.
To do so
we use the idea of defining a potential function $\phi$ and
a function $\psi$ that is an upper bound for $\phi$ and is computable
efficiently. Instead of having one pair of functions $\phi_i,
\psi_i$ for each measure $i$, we now have $\binom{k}{2}$ such functions,
one for each pair of agents. For each measure $i$ and agents $a \ne b$,
define $\phi_{i}^{a, b} = \E{\frac{e^{\lambda X_{a, b, i}} + e^{-
\lambda X_{a, b, i}}}{2}}$, where $X_{a, b, i}$ is the random variable
of the difference between the share of agent $a$ on measure $i$ and
that
of agent $b$ on this measure. The relevant random
distribution here assigns 
every newly created interval to one of the $k$
agents with equal probability which is $1/k$.  The quantity
$g=g(n,k,\epsilon)$ 
is defined here as $g=\frac{\epsilon^2}{100 k \log (nk)}$ and
the function  $\psi_i^{a, b}$ is defined by
$$
\psi_i^{a, b}(t) = \frac{e^{\lambda
x_{t, i}^{a, b}} + e^{- \lambda x_{t, i}^{a, b}}}{2} \cdot e^{2
\lambda^2 g (1 - s_t)/k},
$$ 
where $s_t$ is, as before,
the amount of measure $i$ allocated already, and $x_{t, i}^{a, b}$ is
the discrepancy between $a$ and $b$ on measure $i$ after cut $t$. 

The main difference required here is the replacement of the 
inequality $\cosh(\lambda a) \leq e^{\lambda^2 a^2/2}$ by the 
following inequality which holds whenever, say,
$\lambda a \leq 1$:
$$
\frac{k-2}{k} e^{\lambda \cdot 0} + \frac{1}{k} e^{\lambda a} 
+ \frac{1}{k} e^{-\lambda a} 
=1+\frac{2}{k} (\cosh(\lambda a)  -1)
$$
$$
\leq 1 +\frac{2}{k} (e ^{\lambda^2 a^2/2}-1) 
\leq 1+\frac{2}{k} \frac{2 \lambda^2 a^2}{2} =
1+\frac{2 \lambda^2 a^2}{k} \leq e^{2\lambda^2 a^2/k}.
$$

Each $\phi_i^{a, b}$ is bounded using
the fact that each of the intervals created is of $i$-measure
at most $g $ for every $i$. By the inequality applied with
$a \leq g$ and $\lambda=\frac{\epsilon}{4g}$ (ensuring that indeed
$\lambda a \leq \frac{\epsilon}{4} <1/2$), it follows that if every
interval generated is allocated to an agent in order to minimize 
$\psi = \sum_{a, b \in [k], \enspace a \ne b, \enspace
i \in [n]} \psi_{i}^{a, b}$, then the function $\psi$ never increases
during the algorithm. As
$$
\psi(0)<nk^2 e^{2\lambda^2 g/k} =nk^2 e^{\epsilon^2/8gk} < 
\frac{e^{\lambda \epsilon/k}}{2}
$$
the computation shows that at the end
the absolute discrepancy is $\le \epsilon / k$. We omit the
detailed computation.

The offline algorithm is also adaptable for the general case
of $k$ agents for $\epsilon$-Consensus Splitting. We obtain an algorithm
that makes at most $n (k-1) \lceil 2 + \log_2 \frac{3k}{\epsilon} \rceil$
cuts and provides a proper solution. To this end, we use recursion and
apply a modified version of the algorithm from Theorem \ref{t33} at each
recursion step. Define $\epsilon' = \epsilon / 3k$, and divide the $k$
players into two disjoint groups $A, B$, with $\lfloor k / 2 \rfloor$
agents and $\lceil k / 2 \rceil$ agents respectively. Think of
$A, B$ as two agents and split $[0, 1]$ among them. By following the
algorithm in the proof of Theorem \ref{t33}, one can make 
$\le n (2 + \lceil \log_2
\frac{1}{\epsilon'} \rceil)$ cuts and split the interval so that $A$
gets $\frac{\lfloor k / 2 \rfloor}{k} \pm \epsilon' / 2$ of each measure
$i$. We can do so by starting with all floating coefficients equal to
$\frac{\lfloor k / 2 \rfloor}{k}$ instead of $\frac{1}{2}$ and by following
the proof of Theorem \ref{t33}. Repeat the same procedure for the
groups $A$ and $B$ recursively,  splitting the share of
$A$ 
among its $|A|$ members and doing the same for $B$. It is not
difficult to 
bound the error at the end, checking that it indeed provides  a
proper solution.

Regarding Necklace Splitting, we can adapt the algorithm
for Online $\epsilon$-Consensus Splitting to provide a proper online
solution with $\tilde{O}(n k^{1/3} \cdot m^{2/3} )$ cuts. Note
that this is trivial for $k>m$. For $k \leq m$, define  $\epsilon=
(k/m)^{1/3}$ and proceed as in the proof of Theorem \ref{t34}
defining a color to be critical when the number of remaining beads
of this color is at most $10k^{1/3}m^{2/3}$. As here 
$\epsilon  m_i/k \leq m^{2/3}/k^{2/3}$, when a color becomes
critical the discrepancy between any two agents in it is at most
$ 2 m^{2/3}/k^{2/3}$ and there are enough remaining beads of this
color to enable the algorithn to produce a balanced partition
between all $k$ agents. For the offline
model, we obtain a solution using $O(nk \log m)$ cuts.

Finally we mention that if the number of measures is $n=2$ then for
every $k$ there is an efficient offline algorithm finding a proper solution
with an optimal number of $2k-2$ cuts. This holds for Necklace
Splitting as well as for $\epsilon$-Consensus Splitting. Here
is a sketch for the case of necklaces when the number of beads of
each color is divisible by $k$. Given a necklace with
$m_1$ beads of color $1$ and $m_2$ beads of color $2$ consider it
as a circular necklace. By the discrete intermediate value theorem
there is a circular arc of $(m_1+m_2)/k$ beads containing exactly
$m_1/k$ beads of color 1 (and hence also exactly $m_2/k$ beads of
color $2$.) Cut in the ends of this circular arc, assign it to 
the first agent, and continue inductively.

\subsection{Open questions}
As  stated in an earlier remark, the algorithm in the proof of
Theorem
\ref{t21} gives each agent $\frac{1}{nk}$ of each measure. 
As $n$ tends to infinity, 
this quantity tends to $0$. It will be interesting to decide if the 
$\frac{1}{nk}$ bound can be improved, and in particular, if for
every $k$ there is
a constant $c=c(k) > 0$
so that there exists an efficient algorithm that makes $\le n (k - 1)$
cuts on $[0, 1]$ and gives each agent at least a fraction
$c$ of each measure.\ 

Another open question arises in the context of the Online
$\epsilon$-Consensus Halving problem for $n = 2$ measures, where
the lower bound for the number of cuts is only
$\Omega(\frac{1}{\epsilon})$, whereas the upper bound for the
number of cuts produced by our algorithm  is
$O(\frac{1}{\epsilon^2})$.
The analogous question is open for Online Necklace Halving
with $n = 2$ measures, where the bounds we know for the optimal
number of cuts are
$\Omega(\sqrt{m})$ and $O(m^{2/3})$. \

Lastly, for the general case of $n$ measures for the online version
of $\epsilon$-Consensus Halving there is a logarithmic gap between the
lower bound and the algorithm we provided. For Online Necklace Halving,
the gap is $\Theta((\log n)^{1/3})$. It will be interesting to
close these gaps.

\end{document}